\documentclass{amsart}
\usepackage{amsmath,amsfonts,amscd,amssymb,epsfig,euscript,bm,enumerate}
\usepackage{amsxtra,mathrsfs,xcolor,titletoc} 
\usepackage{mathrsfs}
\usepackage{tikz}
\usepackage{pgf,tikz-cd}
\usetikzlibrary{decorations.markings}
\usepackage{caption,slashed}
\numberwithin{equation}{section}
\usepackage{pdfsync}
\usepackage{bm}
\usepackage{calligra}
\usepackage[T1]{fontenc}
\newtheorem*{theorem*}{Theorem}

\newtheorem{lemma}{Lemma}

\theoremstyle{definition}
{}

\theoremstyle{remark}

\newcommand{\field}[1]{\ensuremath{\mathbb{#1}}}

\newcommand{\RR}{\field{R}}
\newcommand{\TT}{\field{T}}
\newcommand{\ZZ}{\field{Z}}
\DeclareMathOperator{\Tr}{\mathrm{Tr}}

\newcommand{\beq}{\begin{equation}\begin{aligned}}
\newcommand{\eeq}{\end{aligned}\end{equation}}
\newcommand{\eqdef}{\overset{\text{def}}{=}}

\newcommand{\al}{\alpha}

\newcommand{\be}{\beta}
\newcommand{\del}{\delta}

\newcommand{\pa}{\partial}

\newcommand{\ov}{\over}

\newcommand{\e}{\epsilon}

\newcommand{\curly}[1]{\mathscr{#1}}

\newcommand{\cH}{\curly{H}}

\newcommand{\lam}{\lambda}

\newcommand{\ex}[1]{\langle #1 \rangle}

\DeclareMathAlphabet{\mathcalligra}{T1}{calligra}{m}{n}
\DeclareFontShape{T1}{calligra}{m}{n}{<->s*[2.2]callig15}{}

\makeatletter
\@namedef{subjclassname@2020}{\textup{2020} Mathematics Subject Classification}
\makeatother

\usepackage{hyperref}
\hypersetup{colorlinks,citecolor=blue,plainpages=false,hypertexnames=false}

\begin{document}
\title{Supersymmetry and trace formulas III. Frenkel trace formula}
\author{Changha Choi}
\address{Perimeter Institute for Theoretical Physics,
Waterloo, Ontario, N2L 2Y5, Canada}
\author{ Leon A. Takhtajan}
\address{Department of Mathematics,
Stony Brook University, Stony Brook, NY 11794 USA; 
Euler International Mathematical Institute, Pesochnaya Nab. 10, Saint Petersburg 197022 Russia}
\begin{abstract}
By applying the new supersymmetric localization principle introduced in \cite{Choi:2021yuz,Choi:2023pjn}, we present two complementary approaches for the path integral derivation of the `non-chiral' trace formula for a semisimple compact Lie group $G$, which generalizes the so-called Frenkel trace formula. Corresponding physical systems for each picture are the quantum mechanical sigma model on $G$ and the gauged sigma model on $G\times G$, and the approaches closely follow the spirit of the Eskin trace formula \cite{Choi:2021yuz} and the Selberg trace formula \cite{Choi:2023pjn} respectively. These methods provide a natural conceptual bridge between two seemingly independent derivations in \cite{Choi:2021yuz} and \cite{Choi:2023pjn}. 

\end{abstract}
\subjclass[2020]{}
\maketitle
\tableofcontents
\section{Introduction} \label{sec:intro}
 In \cite{Choi:2021yuz,Choi:2023pjn}, we initiated the path integral study of classical trace formulas by formulating a new supersymmetric localization principle. It was applied to the Eskin trace formula \cite{MR0206535}, which deals with the spectrum of the Laplacian on compact Lie groups, and to the Selberg trace formula \cite{selberg1956harmonic},
 which deals with the spectrum of the Laplacian on weakly symmetric Riemannian spaces. 
 
 Although the main ideas, as articulated in \cite{Choi:2021yuz,Choi:2023pjn}, were the same, the physical models and observables used were quite different. Thus in case of the Eskin trace formula, we used  a supersymmetric non-linear sigma model on $G$ and the supertrace with fermionic zero modes inserted as a natural observable, while in case of the Selberg trace formula, we used a supersymmetric gauged sigma model on a quotient of a group manifold by a discrete subgroup, and a natural observable was a singular line operator with fermionic zero modes attached.
 
In the present paper, we consider a `non-chiral' generalization of the Eskin trace formula for a general semisimple  compact Lie group $G$, the so-called Frenkel trace formula.
The methods used for its derivation naturally bridge  our approaches in \cite{Choi:2021yuz} and \cite{Choi:2023pjn}.

Let's recall that the Eskin trace formula deals with the operator trace of the Euclidean evolution operator on the Hilbert space $L^2(G, dg)$,
\beq \label{eq:Eskin}
\Tr \left[L_{g_l} e^{-{1\ov 2}\be {\Delta_G}} \right]=\sum_{\lam \in\text{Irrep}\, G}d_\lam\chi_\lam (g_l) e^{-{1\ov 2}\be C_2(\lam)},
\eeq
where $dg$ is the Haar measure on $G$ is associated with the Riemannian metric determined by the negative of the Cartan-Killing form, and $L_{g_l}$ is a left translation operator on $L^2(G)$, $L_{g_{l}}f(g)=f(g_l g)$. This formulation is in some sense `chiral' since it only involves the left action (of course, one can also formulate it in terms of the right action).
The Eskin trace formula expresses this trace in Lie algebraic terms as a sum over the characteristic lattice $\Gamma=\{\gamma\in\mathfrak t:e^\gamma=1\}$ of $G$, where $\frak{t}$ is a Cartan subalgebra of the Lie algebra $\frak{g}=Lie(G)$. Namely,
\beq \label{eq:Eskin-1}
\Tr \left[L_{g_l} e^{-{1\ov 2}\be {\Delta_G}} \right]=  {\text{vol}(G) e^{{1\ov 2}\be\ex{\rho,\rho}}\ov (2\pi\be)^{n/2}} \sum_{\gamma \in\Gamma}\prod_{\al \in R_+}{{1\ov 2}\ex{\al,h+\gamma}\ov \sinh{1\ov 2}\ex{\al,h+\gamma}} e^{-{\Vert h+\gamma\Vert^{2}\ov 2\be}},
\eeq
where $n=\dim\mathfrak g$, $\rho={1\ov 2}\sum_{\alpha\in R_+}\al$ is the Weyl vector, $g_l = e^{h}$ with regular $h\in\frak{t}$, and $\Vert h+\gamma\Vert^{2}=(h+\gamma,h+\gamma)$, where $(~,~)$ is the inner product on $\frak{g}$, given by the negative of the Cartan-Killing form  (see \cite{Choi:2021yuz} for details).

Physically, the Eskin trace formula deals with the integrated propagator
\beq \label{eq:Eskinphy}
\Tr \left[L_{g_l} e^{-{1\ov 2}\be {\Delta_G}} \right]=\int_{G} \, \langle g_l g | e^{-{1\ov 2}\be \Delta} | g\rangle dg =\text{vol}(G) \langle g_l | e^{-{1\ov 2}\be \Delta} |1\rangle 
\eeq
of the quantum mechanical non-linear sigma model on $G$. Note that the second equality in \eqref{eq:Eskinphy} is due to the bi-invariance of the action.

A remarkable physical interpretation of the Eskin trace formula has been found in \cite{Schulman:1968yv,dowker1970sum,Dowker:1970vu,marinov1979dynamics,Picken:1988ev,Camporesi:1990wm,Choi:2021yuz}. Namely, the right-hand side of \eqref{eq:Eskin-1} is precisely the one-loop contribution around the non-trivial critical points of the action, which are geodesics connecting $g$ and $g_l g$. Moreover, using a new supersymmetric localization principle \cite{Choi:2021yuz},  it is possible to show that the path integral representing \eqref{eq:Eskinphy} localizes to these geodesics. The essential part of the  localization locus is in one-to-one correspondence  with the lattice $\Gamma$, and the exponents in the exponential factors in \eqref{eq:Eskin-1} correspond to the classical action. The elementary prefactors correspond to one-loop determinants, while the exponent involving the Weyl vector corresponds to the DeWitt term, which is obtained  from the supersymmetry algebra \cite{Choi:2021yuz}.

The natural extension of the Eskin trace formula comes from a simple observation that there are other commuting operators that can be inserted into the trace \eqref{eq:Eskin}. Namely, consider the right translation operators $R_{g_r}$ on $L^2(G)$, $R_{g_r}f(g)=f(gg_r)$ and the following `non-chiral' trace
\beq \label{eq:Frenkel bosonic trace}
\Tr \left[L_{g_l} R_{g_r^{-1}} e^{-{1\ov 2}\be {\Delta_G}} \right]=\sum_{\lam \in\text{Irrep}\, G}\chi_\lam (g_l)\chi_\lam (g_r^{-1})e^{-{1\ov 2}\be C_2(\lam)}.
\eeq
It is natural to ask whether a similar trace formula exists. The answer is positive, and such formula, in the case of simply connected, simple compact Lie groups, was derived by I. Frenkel  \cite[Theorem 4.3.4]{Frenkel1984}  by combining the Weyl character formula and the Poisson summation formula. The Frenkel trace formula has the following form (see also \cite{bismut2008hypoelliptic})
\beq \label{eq:Frenkeltf}
\begin{gathered}
\Tr \left[L_{g_l} R_{g_r^{-1}} e^{-{1\ov 2}\be {\Delta_G}} \right] \\ ={\text{vol}(\TT )e^{{1\ov 2}\be \ex{\rho,\rho}}\ov (2\pi \be)^{r \ov 2}\mathfrak s (h_l)\mathfrak s (-h_r)}\sum _{ (w,\gamma)\in W \times 2\pi i Q^\vee }\epsilon(w) e^{-{\Vert h_l-w h_r+\gamma\Vert^{2}\ov2 \be}},
\end{gathered}
\eeq
where  $\TT\subset G$ is the maximal torus, $r=\text{dim}\,\TT$, and $g_{l,r}=e^{h_{l,r}}$ with regular $h_{l,r}\in \mathfrak t$. 
Furthermore, $Q^\vee$ is the coroot lattice, $W$ is the Weyl group and $\mathfrak s(h)$ is a denominator of the Weyl character formula 
\beq \label{eq:Weylden}
\mathfrak s (h)\eqdef \prod_{\al\in R_+}2\sinh{\ex{\al,h}\ov 2}.
\eeq

Having in mind the rich history of physical interpretation of the Eskin trace formula, one can ask whether the Frenkel trace formula \eqref{eq:Frenkeltf} can be interpreted as a  semiclassical sum over critical points. For this aim we observe that  similar to  \eqref{eq:Eskinphy}, the operator trace in \eqref{eq:Frenkeltf} 
can be represented as the  integrated propagator of a quantum mechanical non-linear sigma model on $G$,
\beq \label{eq:Frenkelphy}
\Tr\left[L_{g_l}R_{g_r^{-1}} e^{-{1\ov 2}\be \Delta_G}\right]=\int_{G}  \langle  g_l g g_r^{-1} |e^{-{1\ov 2}\be \Delta_G} | g\rangle dg.
\eeq

However, there are significant differences between Frenkel and Eskin trace formulas. First, the integrated propagator in \eqref{eq:Frenkelphy} cannot be simplified as in \eqref{eq:Eskinphy}. Second, the semiclassical saddle points of the integrand in \eqref{eq:Frenkelphy} for each $g$ are geodesics connecting $g$ and $g_l g g_r^{-1}$, and each of them has a semiclassical action proportional to the geodesic  length square, which non-trivially depends on $g$. However, the right-hand side of the Frenkel trace formula \eqref{eq:Frenkeltf} has an exponent which does not depend on $g$ at all. Moreover, the sum in the Frenkel trace formula additionally goes over the Weyl group $W$ (note that $\Gamma=2\pi iQ^\vee$ for simply connected $G$), whose semiclassical interpretation in terms of the geodesics on $G$ is not clear. 

Here we present two complementary approaches to semiclassical  interpretation of the Frenkel trace formula, which are exact and obtained from our new supersymmetric localization principle.   As a result, we get a general form of the  Frenkel trace formula, valid for an arbitrary semisimple compact Lie group. 

The first approach is  based on the non-linear sigma model on $G$, used in the proof of Eskin trace formula  in \cite{Choi:2021yuz}. Suitably modifying the framework by incorporating the right twist $R_{g_r^{-1}}$, we show that the localization formalism worked out in \cite{Choi:2021yuz} is adaptable in this case.

The second approach is based on the interpretation of  the non-linear sigma model on $G$ as a gauged non-linear sigma model on $G\times G$ using the isomorphism
\beq 
 (G\times G)/G\simeq G.
\eeq
The derivation of the Frenkel trace formula using this approach becomes quite analogous to the derivation of the Selberg trace formula in \cite{Choi:2023pjn} by localization.  We will see how many novel concepts introduced in \cite{Choi:2023pjn} are naturally applied here. 

The main mathematical input in both approaches is the Harish-Chandra formula \cite{harish1957differential} (see also \cite[Chapter 7.5]{berline2003heat})
\beq \label{eq:HC}
\int_{G/\TT} e^{i\ex{X,g^{-1}\lam g}} dg= \left(\prod_{\al\in R_+} {2\pi \ov   i \ex{\al,X} \ex{\al, \lam }}  \right)  \sum_{w\in W}\e(w) e^{i\ex{w\lam,X}},
\eeq
where $X,\lam\in \mathfrak t$ are regular, and the Haar measure $dg$ on $G$ is associated with the Riemannian metric on $G$, determined by the negative of the Cartan-Killing form. It identifies the orbital integral --- Fourier transform of a Liouville volume form on a coadjoint orbit --- with the sum over the Weyl group.

The organization of the paper is the following. In Section \ref{sec:sm} we prove the Frenkel trace formula using our new supersymmetric localization principle, applied to the non-linear sigma model on $G$ as in  \cite{Choi:2021yuz}. In Section \ref{sec:gsm} we follow \cite{Choi:2023pjn} and present another derivation, that uses the gauged sigma model on $G\times G$. These two approaches bridge our papers \cite{Choi:2021yuz} and \cite{Choi:2023pjn}, and these two sections can be read independently of each other. 

\subsection{Acknowledgments} The research of C.C. was supported by the Perimeter Institute for Theoretical Physics. Research at Perimeter Institute is supported in part by the Government of Canada through the Department of Innovation, Science and Economic Development and by the Province of Ontario through the Ministry of Colleges and Universities. The second author (L.T.) thanks E. Witten for raising the question of the non-chiral extension of the Eskin trace formula, and I. Frenkel for pointing to the reference \cite{WENDT200165}.

\section{Supersymmetric sigma model on $G$} \label{sec:sm}

We start by recalling necessary basic facts about the supersymmetric nonlinear sigma model on a compact semisimple Lie group $G$ (see \cite{Choi:2021yuz} for the detailed exposition and notations).
The total Hilbert space of the model is 
$$\curly H=L^2(G)\otimes \curly H_{F,\mathfrak g},$$ 
where $\curly H_{F,\mathfrak g}$ is the Hilbert space of $n$ free Majorana fermions (irreducible Clifford algebra module of dimension $2^{n/2}$). Let $\frak{g}$ be the Lie algebra of $G$ with a basis $\{T_{a}\}$, where $[T_a,T_b]=f_{ab}^cT_c$, and let $g_{ab}=-B(T_a,T_b)$,  where $B$ is the Killing form.
The supercharge $\hat{Q}$ is the following operator on $\cH$, 
\beq
\hat Q=\hat\psi^a \hat l_a+{i\ov 6}f_{abc}\hat\psi^a\hat\psi^b \hat\psi^c,
\eeq
where $\hat l_a$ are first order differential operators on $L^{2}(G)$ generated by the left-invariant vector fields on $G$, and $\hat \psi^a$ are Majorana fermions. They satisfy the following commutation relations
\beq
\quad [\hat{l}_{a},\hat{l}_{b}]=-if_{ab}^c\hat l_c, \quad [\hat \psi^a,\hat \psi^b]=g^{ab}.
\eeq

The supersymmetric quantum Hamiltonian is given by
\beq
\hat H=\hat Q^2={1\ov2} {\Delta}+{R\ov 12}\hat I,
\eeq
where $\Delta=g^{ab}\hat l_a \hat l_b$ is the Laplace operator on $G$, and $R$ is the scalar curvature of  the Cartan-Killing metric on $G$. Using the Freudenthal-de Vries' `strange formula', we have $R=6{\ex{\rho,\rho}}$, where $\rho={1\ov 2}\sum_{\alpha\in R_+}\al$ is the Weyl vector.

The main task is to identify the localizable supertrace, which is directly related to the bosonic trace \eqref{eq:Frenkel bosonic trace} of the Frenkel trace formula. A generic form is
\beq\label{O-F}
\text{Str}\left[\mathcal O L_{g_l} R_{g_r^{-1}} e^{-\be \hat H}\right]
\eeq
with the insertion of some operator $\mathcal O$, which we will identify later. As was explained in the Section \ref{sec:intro}, we can assume that $g_l,g_r\in\TT$, so $g_{l,r}=e^{h_{l,r}}$ where $h_{l,r}\in \mathfrak t$.

Also,  since we extended the bosonic Hilbert space $L^{2}(G)$ to the Hilbert space $\cH$, we need to specify how operators $L_{g_l}$, $R_{g_r^{-1}}$ in \eqref{O-F} act on the fermionic Hilbert space $\curly H_{F,\mathfrak g}$. It is natural to require that operators $L_{g_l}$, $R_{g_r^{-1}}$ and $\hat Q$ commute. This can be achieved by the following natural action
\beq \label{eq:LRactionHF}
\quad [L_{g_l},\hat\psi ]=0,\quad [R_{g_r^{-1}},\hat \psi]=g_r \hat \psi g_r^{-1},\quad\text{where}\quad \hat\psi=\psi^a T_a.
\eeq

There are two natural ways of interpreting the insertion of operators $L_{g_l}$, $R_{g_r^{-1}}$ into \eqref{O-F} in terms of the path integral formulation of the supertrace. One is to change periodic boundary conditions to twisted boundary conditions while keeping the action associated with $\hat H$, and the other is to modify the action to a twisted one while keeping the periodic boundary conditions. Here,  as in \cite{Choi:2021yuz}, we adopt the latter point of view.

Namely, we note that $L_{g_l}$ and $R_{g_r^{-1}}$ are unitary operators of the left and the right global $G$-symmetries. Therefore, the twisted action of the model can be understood as a result of coupling to the background gauge fields $A_l={h_l\ov \be}$ and $A_r={h_r\ov \be}$. This straightforwardly determines the twisted Euclidean action as 
\beq \label{eq:twistedS}
~&S^{h_l,h_r}_E=\int_0^\be  \left( {1\ov2}(J^{l,r},J^{l,r})+{1\ov 2}(\psi,(\pa_t+\text{ad}_{h_r/\be})\psi) \right) d\tau ,
\\ & J^{l,r}= J+\text{Ad}_{g^{-1}}{h_l\ov \be} -{h_r\ov \be},\quad J=g^{-1}\dot g,
\eeq
where $(\, ,\, )$ is a bilinear form on $\mathfrak g$ given by the Cartan-Killing metric.

Moreover, since $L_{g_l}$, $R_{g_r^{-1}}$ and $\hat Q$ are mutually commuting, the twisted action possesses a supersymmetry, which is manifested by the transformations
\beq\label{susy-1}
~&\del g= g\psi
\\& \del \psi=-J^l-\psi\psi,
\eeq
where $J^l=J+\text{Ad}_{g^{-1}}{h_l\ov \be}$. 

It is really nice that this supersymmetry is exactly the same as the one used in our derivation of the Eskin trace formula \cite{Choi:2021yuz}. Namely, it 
does not depend on  $h_r$, because our supersymmetric particle on $G$ is endowed with the left-invariant connection with totally antisymmetric parallelizing torsion.

For regular $h_r$, the twisted action \eqref{eq:twistedS} has $r=\text{dim}\,\TT$ fermionic zero modes,
\beq
\chi^i={1\ov \be}\int_0^\be \psi^i d\tau, \quad   T_i \in \mathfrak t,\quad i=1\dots r,
\eeq
where we expand $\psi$ in terms of the Cartan-Weyl basis \cite{Choi:2021yuz}, so $T_j=iH_j$, where $H_j\in i\mathfrak t$ and $\ex{H_i,H_j}=\del_{ij}$.

Choosing the operator $\mathcal{O}$ in \eqref{O-F}  to be $ c_r \hat \chi^1 \cdots \hat \chi^{r}$ where\footnote{The overall sign ambiguity in the definition of $c_r$ is fixed by the definition of the fermion measure $\mathscr D\psi$ in the path integral;  see \cite{Choi:2021yuz} for details.} $c_r=\pm i^{r(r-1)/2}$, we get the following supertrace 
\beq\label{str-1}
I=\text{Str}_{\curly H}\left[ c_r \hat \chi^1 \cdots \hat \chi^{r}L_{g_l}R_{g_r^{-1}}e^{-\be \hat H}\right]=\bm\int_{\Pi TLG}  \chi^1 \cdots \chi^r e^{-S^{h_l,h_r}_E} \mathscr D  g \mathscr D\psi,
\eeq
where $\Pi TLG$ denotes the configuration space of field configurations
$(g,\psi)$ consisting of a loop $g\in LG$ together with a Grassmann-odd
tangent vector $\psi\in T_g LG \simeq \Gamma(S^1_\beta,g^*TG)$.
Here $LG \equiv \mathrm{Map}(S^1_\beta,G)$ denotes the free loop group of $G$,
i.e.\ the space of smooth maps $g(\tau):S^1_\beta\to G$.

It corresponds to the operator part of the Frenkel trace formula, since
\beq \label{eq:SGsupertrace}
\begin{gathered}
\text{Str}_{\curly H}\left[ c_r \hat \chi^1 \cdots \hat \chi^{r}L_{g_l}R_{g_r^{-1}}e^{-\be \hat H}\right]
\\
=e^{-{1\ov 2}\be 
\ex{\rho,\rho}}\Tr_{L^2(G)}\left[L_{g_l}R_{g_r^{-1}}  e^{-{1\ov 2}\be \Delta_G }\right] 
\text{Str}_{\curly H_{F,\mathfrak g}}\left[ c_r
\hat \chi^1 \cdots \hat \chi^{r}R_{g_r^{-1}}\right],
\end{gathered}
\eeq
where $L_{g_l}$ is dropped out of the supertrace because of \eqref{eq:LRactionHF}. The fermionic trace can be computed as in \cite{Choi:2023pjn}, which gives
\beq \label{eq:SGfermionint}
\text{Str}_{\curly H_{F,\mathfrak g}}\left[ c_r
\hat \chi^1 \cdots \hat \chi^{r}R_{g_r^{-1}}
\right]=\prod_{\al \in R_+} -2 \sinh \frac{\ex{h_r,r}}{2}=\mathfrak s (-h_r). 
\eeq
Now it is straightforward to see that the path integral in \eqref{str-1} can be computed by the new supersymmetric localization principle in \cite{Choi:2021yuz}, where the localizing deformation is the same as in our derivation of the Eskin trace formula. Namely, 
\beq
~&V=-{1\ov 2}\int_0^\be  (   \dot J^l, \dot \psi) d\tau, 
\\& \del V={1\ov 2}\int _0^\be \left( (\dot J^l, \dot J^l)+(\dot \psi,(\pa_\tau+\text{ad}_{J^l})\dot \psi)  \right) d\tau,
\eeq
so the remaining task is to evaluate the localized path integral
\beq
I=\lim_{s\rightarrow \infty}\bm\int_{\Pi TLG}  \chi^1 \cdots \chi^r e^{-S^{h_l,h_r}_E-s\del V} \mathscr D  g \mathscr D\psi.
\eeq

The localization locus is the subset in $LG$, defined by  $\dot J^l=0$. Similar to \cite{Choi:2021yuz}, it follows from the periodic boundary conditions $g(0)=g(\be)$  that solutions are
\beq \label{eq:SGloc}
g(\tau)=e^{{1\ov \be}\gamma \tau }g_0, \quad \gamma\in \Gamma,
\eeq
where $g(0)=g_{0}\in G$ is the initial condition, and $\Gamma=\{\gamma\in\mathfrak t: e^\gamma=1\}$ is the characteristic lattice. Therefore after localization, the infinite-dimensional bosonic integration domain $LG$ reduces to the set \eqref{eq:SGloc}, parametrized by $g_0\in G$ and $\gamma\in \Gamma$. The corresponding classical action is given by
\beq
S^{h_l,h_r}_E(g_0,\gamma)=\frac{1}{2\beta}\left(\Vert h_1 + \gamma\Vert^{2}-2(h_l+\gamma,g_0 h_r g_0^{-1})+\Vert h_r\Vert^{2}\right).
\eeq

The computation of the one-loop functional determinants is identical to the one performed in \cite{Choi:2021yuz}, so we obtain
\beq\label{1-loop-1}
\begin{gathered}
{\text{Pf}(-\pa_\tau^3-\text{ad}_{(h_l+\gamma)/\be } \pa_\tau^2)|_{L \mathfrak g/\mathfrak g}\ov 
\det(-\pa_\tau^2-\text{ad}_{(h_l+\gamma)/\be}\pa_\tau)|_{L \mathfrak g/\mathfrak g}
}={1\ov \be^{n\ov  2}} \prod_{\al \in R_+}{{1\ov 2}\ex{\al,h_l+\gamma } \ov   i \sinh{1\ov 2}\ex{\al,h_l+\gamma } }  \\
={1\ov \be^{n\ov  2} } {\pi(h_l+\gamma)  \ov  \mathfrak s(h_l+\gamma
)},
\end{gathered}
\eeq
where $\mathfrak s(h)$ was defined in \eqref{eq:Weylden}, and we have put 
\beq
\boldsymbol{\pi}  (h)\eqdef \prod_{\al \in R_+} -i\ex{\al,h}.
\eeq

The remaining finite-dimensional fermionic integral over the space $\Pi  (\mathfrak g/\mathfrak t) \subset \Pi L\mathfrak   g$ of constant modes $\psi_{\mathfrak g/\mathfrak t}$ is easily computed to be
\beq \label{eq:zeroint1}
\int_{\Pi (\mathfrak g/\mathfrak t)} e^{-\int_0^\be  {1\ov 2} (\psi_{\mathfrak g/\mathfrak t},(\pa_t+\text{ad}_{h_r/\be})\psi_{\mathfrak g/\mathfrak t} ) d\tau } d\psi_{\mathfrak g/\mathfrak t}  =\prod_{\al \in R_+}-i \ex{\al,h_r}=\boldsymbol\pi (h_r).
\eeq
Thus the localized path integral $I$ in \eqref{O-F} takes the form
\beq\label{int-loc-1}
\begin{gathered}
I=\sum_{\gamma \in \Gamma}\int_G  {\boldsymbol \pi (h_l+\gamma )\boldsymbol\pi (h_r) \ov (2\pi \be)^{n\ov  2} \mathfrak s(h_l+\gamma)} e^{-{\Vert h_l +\gamma\Vert^{2}-2(h_l+\gamma,g_0 h_r g_0^{-1})+\Vert h_r\Vert^{2}\ov 2\be}}\,  dg_0
\\= \sum_{\gamma \in \Gamma} {\boldsymbol \pi (h_l+\gamma )\boldsymbol\pi (h_r) \ov (2\pi \be)^{n\ov  2} \mathfrak s(h_l+\gamma)}  e^{-{\Vert h_l+\gamma\Vert^{2}+\Vert h_r\Vert^{2}\ov  2\be }}\int_G e^{{(h_l+\gamma,g_0 h_r g_0^{-1})\ov \be}}\,dg_0
\\= \text{vol}(\TT) \sum_{\gamma \in \Gamma} {\boldsymbol \pi (h_l+\gamma )\boldsymbol\pi (h_r) \ov (2\pi \be)^{n\ov  2} \mathfrak s(h_l+\gamma)} e^{-{\Vert h_l+\gamma\Vert^{2}+\Vert h_r\Vert^{2}\ov  2\be }}\int_{G/\TT} e^{{(h_l+\gamma,g_0 h_r g_0^{-1})\ov \be}}\,dg_0. 
\end{gathered}
\eeq

The orbital integral in \eqref{int-loc-1} can be evaluated using the analytic continuation $X\rightarrow -iX$ of the Harish-Chandra formula \eqref{eq:HC} for compact $G$,
\beq \label{eq:HCAC}
\int_{G/\TT} e^{(X,g^{-1}\lam g)} dg= {(2\pi)^{{1\ov 2}\text{dim}(G/\TT)} \ov \boldsymbol \pi (X) \boldsymbol\pi(\lam )}   \sum_{w\in W}\e(w) e^{(w\lam,X)},\quad X,\lambda\in\frak{t}.
\eeq
Thus we obtain
\beq
I= \sum_{(w,\gamma) \in W\times \Gamma} {\text{vol}(\TT)\ov (2\pi \be)^{r \ov  2} \mathfrak s(h_l+\gamma ) }  \e(w) e^{-{1\ov 2\be}\Vert h_l-wh_r+\gamma\Vert^{2}}
\eeq
where $r=\text{dim}\,\TT$. Together with \eqref{eq:SGsupertrace} and \eqref{eq:SGfermionint}, we get
\beq
\begin{gathered}
\Tr_{L^2(G)}\left[L_{g_l}R_{g_r^{-1}}  e^{-{1\ov 2}\be \Delta_G }\right] \\={\text{vol}(\TT) e^{{1\ov 2}\be\ex{\rho,\rho}} \ov (2\pi \be)^{r\ov 2}  } \sum_{(w,\gamma) \in W\times \Gamma} {\e(w) e^{-{\Vert h_l-wh_r+\gamma\Vert^{2}\ov 2\be}} \ov \mathfrak s(h_l+\gamma )  \mathfrak s(-h_r) }.
\end{gathered}
\eeq
At first glance,  the denominator makes this formula look woefully asymmetric with respect to  $h_l$ and $h_r$. However, we have the following basic fact.

\begin{lemma} \label{eq:pairingId1} Let $G$ be a compact semi-simple Lie group, and $\Gamma$ be its characteristic lattice. Then
$e^{\ex{\al, \gamma }}= 1$ for all $\al\in R$ and $\gamma \in \Gamma$.
\end{lemma}
\begin{proof} For each $\al \in R$, there is a natural 1-dimensional representation of $\TT$ acting on $ \mathfrak g_{\al}$  by the adjoint action,
$$
\text{Ad}_{\bm t} X= e^{-\al(\bm h)} X,\quad \bm t=e^{\bm h} \in T,\quad X\in \mathfrak g_\al,
$$
which immediately follows from $\text{Ad}\circ \exp=  \exp\circ\,\text{ad}$. Since this equality is true for any $\bm h \in \mathfrak t$ satisfying $\bm t=e^{\bm h}$, we necessarily have $\al(\gamma)=\ex{\al,\gamma}\in 2\pi i \mathbb Z$ for all $\al\in R$ and $\gamma\in \Gamma$.
\end{proof}
From Lemma \ref{eq:pairingId1} we have $e^{{1\ov 2}\ex{\al, \gamma }}= \pm1$  which gives
\beq \label{eq:WeyldenId}
\mathfrak s(h_l+\gamma)=\mathfrak s (h_l) e^{\ex{\rho,\gamma} }=\mathfrak s (h_l) e^{-\ex{\rho,\gamma} },
\eeq
where $\rho={1\ov 2}\sum_{\al\in R_+}\al$ is the Weyl vector. 

As a result, we have obtained the generalization of the Frenkel trace formula for a general compact semisimple Lie group $G$,
\beq \label{eq:Frenkel trace formula final}
\begin{gathered}
  \Tr_{L^2(G)}\left[L_{g_l}R_{g_r^{-1}}  e^{-{1\ov 2}\be \Delta_G }\right]
\\={\text{vol}(\TT) e^{{1\ov 2}\be\ex{\rho,\rho}} \ov (2\pi \be)^{r\ov 2} \mathfrak s(h_l )  \mathfrak s(-h_r)   } \sum_{(w,\gamma) \in W\times \Gamma} \e(w) e^{\ex{\rho,\gamma}}e^{-{\Vert h_l-wh_r+\gamma\Vert^{2}\ov 2\be}} .
\end{gathered}
\eeq

In the special case when $G$ is simply connected, $\Gamma =2\pi iQ^\vee $ and  $e^{\ex{\rho,\gamma}}=1$ for all $\gamma\in \Gamma$, so our formula \eqref{eq:Frenkel trace formula final} reduces to the formula \eqref{eq:Frenkeltf}, derived by Frenkel \cite{Frenkel1984} (see also \cite{WENDT200165,bismut2008hypoelliptic}).

\section{Supersymmetric gauged sigma model on $G\times G$ }\label{sec:gsm}
As was briefly mentioned in Section \ref{sec:intro}, we can alternatively view $G$ in terms of the symmetric space
\beq \label{eq:Giso}
 (G\times G)/G\simeq G,
\eeq
where the quotient identifies $(g_1,g_2)\sim (g_1g,g_2g)$ and the isomorphism is given by the mapping $(g_1,g_2) \mapsto g_1g_2^{-1}$. 

In view of \eqref{eq:Giso}, we can equivalently represent the non-linear sigma model on $G$ as the gauged sigma model on $G\times G$, obtained by gauging the subgroup $G$ with the gauge symmetry given by 
diagonal right $G$-action on $G\times G$.

To write down the corresponding quantum mechanical model, following \cite{Choi:2021yuz} we first introduce the bosonic non-linear sigma model on the group $G\times G$ with the action 
\beq
\mathcal L^B_{G\times G}=\ex{J_1,J_1}+\ex{J_2,J_2},
\eeq
where $g_{p}$ are $G$-valued fields and $J_p=g^{-1}_p\dot{g}_p$, $p=1,2$, are left-invariant currents. This system has a Hamiltonian $\hat H^B_{G\times G}={1\ov 4}(\Delta_1+\Delta_2)$ acting on the Hilbert space $\mathscr H_{G\times G}=
L^2(G)\otimes L^2(G) $, and its spectrum is parameterized by $\text{Irrep } G\times \text{Irrep } G$.

As in \cite{Choi:2023pjn}, the quotient by $G$ can be physically realized by introducing a $\mathfrak g$-valued gauge field $A$ which is a connection on a principal $G$-bundle $P$ over a $0+1$ dimensional `spacetime' $S^{1}_{\beta}=\RR/\beta\ZZ$. This implements the following gauge symmetry
\beq
g_p\mapsto g_pg(t),\quad A\mapsto g^{-1}(t)Ag(t)+g^{-1}(t)\dot g(t), \quad p=1,2.
\eeq
Then the action of the bosonic gauged sigma model on $(G\times G)/G$ is
\beq \label{eq:bosonicgsm}
\mathcal L^B_{(G\times G)/G}=\ex{J_{1,A},J_{1,A}}+\ex{J_{2,A},J_{2,A}},
\eeq
where $J_{p,A}=J_p-A$ are  covariant currents, $p=1,2$.

Consider the operators $\hat{l}_{a}^1=\hat{l}_{a}\otimes I$ and $\hat{l}_{a}^2=I\otimes \hat{l}_{a}$ on $L^{2}(G\times G)=L^{2}(G)\otimes L^{2}(G)$,  where $\hat{l}_{a}$ were defined in Section \ref{sec:sm}. The quantization of the Gauss law 
$$J_{1,A}+J_{2,A}=0$$
leads to the constraints  
$$\hat{l}_{a}^1+\hat{l}^{2}_a=0,\quad a=1,\dots,n$$ 
on the states in $\mathscr H_{G\times G}=L^{2}(G\times G)$, and defines the physical Hilbert space $\mathscr H_{G\times G/G}\simeq L^{2}(G)$. Then it is straightforward to see that the physical spectrum of the Hamiltonian is isomorphic to $\text{Irrep }G$ with ${1\ov 2}C_2(R)$ as corresponding energy eigenvalues, therefore giving a physical interpretation of the isomorphism \eqref{eq:Giso}.

The corresponding supersymmetric gauged sigma model can be constructed in the spirit of \cite{Choi:2023pjn} where the Lagrangian is
\beq \label{eq:GSM Lag}
\mathcal L_{(G\times G)/G}=\ex{J_{1,A},J_{1,A}}+\ex{J_{2,A},J_{2,A}}+{i} \ex{\psi_1,D_A\psi_1}+{i} \ex{\psi_2,D_A\psi_2},
\eeq
which possess a supersymmetry
\beq 
~&\del g_p=ig_p \psi_n
\\&\del \psi_p=-J_{p,A}-i\psi_p\psi_p \\&\del A=0.
\eeq

Now we need to implement the action of $L_{g_l}R_{g_r^{-1}}$ in the gauged sigma model picture. The advantage of the coset formulation is that the `non-chiral' twist $L_{g_l}R_{g_r^{-1}}$ on $G$ becomes purely a left twist $L^{(1)}_{g_l} L^{(2)}_{g_r}$ acting on $G\times G$, where $L^{(p)}_g$ acts on the $p$-th factor of $G$, $p=1,2$. In terms of adding a background gauge field, what we need is a replacement
\beq
J_{p,A}\rightarrow J_{p,A}^{h_p}=J_{p,A}+{1\ov \be}g_p^{-1}h_pg_p,
\eeq
where for notational simplicity we have introduced $h_{1}=h_l$, $h_2=h_r$, so $g_l=e^{h_1},g_r=e^{h_2}$. The corresponding twisted action in the Euclidean signature becomes \
\beq \label{eq:SEgsm}
S'_E=\int_{0}^{\beta} \sum_{p=1,2}\left((J_{p,A}^{h_p},J_{p,A}^{h_p})+ (\psi_p,D_A\psi_p)\right)d\tau,
\eeq
and has a Euclidean supersymmetry
\beq \label{eq:susygsm}
~&\del_h g_p=g_p\psi_p,
\\&\del_h \psi_p=-J_{p,A}^{h_p}-\psi_p\psi_p,\quad p=1,2.
\eeq

This setup is quite analogous to \cite{Choi:2023pjn}, where bosons and fermions were coupled through the gauge field. Namely,  for regular $A$ the action \eqref{eq:SEgsm} has $2r=r+r$ fermion zero modes,  coming from the kernel of $\nabla_A=d+\text{ad}_A$ for each $\psi_{p}$, $p=1,2$. 
For each $\psi_p$ denote by $\chi_p^{1},\dots,\chi_{p}^{r}$ the zero modes and put 
\beq \label{eq:GSMfermzero}
\chi_p(A)=c_r 2^{r/2}\chi_p^1\dots \chi_p^r,
\eeq
where now the extra factor of $2^r$ appears because of the normalization of the fermionic Lagrangian in \eqref{eq:GSM Lag}. The normalization is fixed by the requirement that these zero modes are associated with the orthonormal basis of $\mathfrak t$ in the temporal gauge $\dot A=0$ and $A\in \mathfrak t$. Then we have the following identity, relating the bosonic trace in terms of a path integral associated with the gauged sigma model action $S'_E$ (cf.  \cite[Section 5]{Choi:2023pjn}),
\beq \label{eq:GSMmain}
\begin{gathered}
\Tr_{L^2(G)} \left[L_{g_l} R_{g_r^{-1}} e^{-{1\ov 2}\be {\Delta_G}} \right] 
\\ ={e^{{1\ov 2}\be\ex{\rho,\rho}} \ov \text{vol}(\mathcal G)}\bm\int  {\chi_1(A)\chi_2(A) e^{- S'_E}\ov  \left[ \text{Pf}\,'\left( i(\text{Hol}_{S_\be^1}^{-1/2}(\nabla _A)  - \text{Hol}_{S_\be^1}^{1/2}(\nabla _A)  )\right) \right]^2 }\prod_{p=1,2}\mathscr D g_p \mathscr D \psi_p \mathscr DA.
\end{gathered}
\eeq
Here the denominator comes from the fermionic path integral, and $e^{{1\ov 2}\be \ex{\rho,\rho}}$ is a DeWitt term that naturally comes from the supersymmetry algebra as in \cite{Choi:2021yuz,Choi:2023pjn}.

Finally,  following closely \cite{Choi:2023pjn}, we localize the right-hand side of \eqref{eq:GSMmain}. First, we gauge fix $A$ to the temporal gauge $A= h/\be=\text{const}$ and  use the residual gauge symmetry to effectively reduce the integration over $A$ to the integration along the holonomy $\bm t=e^{ h}\in \TT$ as
\beq \label{eq:gaugefixed1}
\begin{gathered}
\Tr_{L^2(G)} \left[L_{g_l} R_{g_r^{-1}} e^{-{1\ov 2}\be {\Delta_G}} \right] 
\\ ={e^{{1\ov 2}\be\ex{\rho,\rho}} \ov |W|\text{vol}(\TT)} \int_{\TT}   |\del(\bm t)|^2 \left(\bm\int  {\chi_1(h/\be)\chi_2( h /\be) e^{- S'_E}\ov \mathfrak s( h)^2 }\prod_{p=1,2}\mathscr D g_p \mathscr D \psi_p\right)d\bm t,
\end{gathered}
\eeq
where
\beq
\del(\bm t)=\det(1-\text{Ad}_{\bm t})_{\mathfrak g/\mathfrak t}=\prod_{\al\in R_+}\left(e^{\ex{\al,h}\ov 2}-e^{-{\ex{\al,h}\ov 2}} \right),\quad \bm{t}=e^{h}.
\eeq
Therefore $|\delta(\bm t)|^2=\mathfrak s(h)\mathfrak s(-h)=\mathfrak s(h)^2 (-1)^{\text{dim}(G/\TT)\ov 2}$, and hence \eqref{eq:gaugefixed1} is simplified as
\beq \label{eq:gaugefixed2}
\begin{gathered}
\Tr_{L^2(G)} \left[L_{g_l} R_{g_r^{-1}} e^{-{1\ov 2}\be {\Delta_G}} \right] 
\\=(-1)^{\text{dim}(G/\TT)\ov 2} {e^{{1\ov 2}\be\ex{\rho,\rho}} \ov |W|\text{vol}(\TT)} \int_{\TT}  \left( \bm\int  {\chi_1(h/\be)\chi_2( h/\be) e^{- S'_E} }\prod_{p=1,2}\mathscr D g_p \mathscr D \psi_p\right)d\bm{t}.
\end{gathered}
\eeq
 In accordance with the new localization principle \cite{Choi:2021yuz,Choi:2023pjn}, the path integral in the r.h.s. of \eqref{eq:gaugefixed2} admits the following invariant deformation $S'_E\rightarrow S'_E+\lambda \del_h (V_1+V_2)$ with
\beq
\begin{gathered}
V_p=-\int_0^\beta (\dot  J_{p}^{h_p},  \dot \psi_p) d\tau=-\int_0^\beta (\dot  J_{p,h/\be}^{h_p},  \dot \psi_p) d\tau, \\
\del V_p=\int _0^\be \left( (\dot J_{p}^{h_p}, \dot J_{p}^{h_p})+(\dot \psi,(\pa_\tau+\text{ad}_{J^{h_p}_p})\dot \psi)  \right) d\tau,
\end{gathered}
\eeq
where $J_p^{h_p}=J_p+{1\ov \be}g_p^{-1}h_p g_p$, $p=1,2$.

The localization locus is determined by the equation $\dot J_p^{h_p}=0$ whose solutions are
\beq \label{eq:locus}
g_p(\tau)=e^{{1\ov \be}\gamma_p \tau }g_{p,0}, \quad \gamma_p\in \Gamma.
\eeq
It is parametrized by $(\gamma_1,\gamma_2)\in \Gamma\times \Gamma$ and the corresponding classical action is 
\beq \label{eq:Scl}
\begin{gathered}
S'_E\left
(\{g_{p,0}\},\{\gamma_p\}, A
={h\ov \be}\right)\\=\frac{1}{\beta}\sum_{p=1,2} \left(\Vert h_p+\gamma_p\Vert^{2}-2(h_p+\gamma_p,g_{p,0} h g_{p,0}^{-1})+\Vert h\Vert^{2}\right).
\end{gathered}
\eeq
As in \eqref{1-loop-1}, the-loop determinant is given by
\beq
\prod_{p=1,2}{ \text{Pf}(-\pa_\tau^3-\text{ad}_{(h_p+\gamma_p)/\be } \pa_\tau^2)|_{L \mathfrak g/\mathfrak g}\ov 
\det(-\pa_\tau^2-\text{ad}_{(h_p+\gamma_p)/\be}\pa_\tau)|_{L \mathfrak g/\mathfrak g}
}&={1\ov \be^{n} } \prod_{p=1,2}{\boldsymbol {\pi}(h_p+\gamma_p)  \ov  \mathfrak s(h_p+\gamma_p
)},
\eeq
so using \eqref{eq:zeroint1}, we get for the remaining fermionic integral
\beq
\prod_{p=1,2}\int_{\Pi (\mathfrak g/\mathfrak t)} e^{-\int_0^\be   (\psi_{p,\mathfrak g/\mathfrak t},(\pa_t+\text{ad}_{h/\be})\psi_{p,\mathfrak g/\mathfrak t} ) d\tau } d\psi_{\mathfrak g/\mathfrak t} =2^{\text{dim}(G/\TT)}\boldsymbol\pi (h)^2.
\eeq

Combining these results, we see that localization reduces the path integral into the following orbital integral
\beq
\begin{gathered}
\Tr_{L^2(G)} \left[L_{g_l} R_{g_r^{-1}} e^{-{1\ov 2}\be {\Delta_G}} \right] = {(-1)^{\text{dim}(G/\TT)\ov 2}\text{vol}(\TT) 2^{\,r+\dim(G/\TT)} e^{{1\ov 2}\be\ex{\rho,\rho}} \ov |W|(2\pi\be) ^n} 
\\\times \sum_{\gamma_1,\gamma_{2}\in\Gamma}{\boldsymbol {\pi}(h_1+\gamma_1)  \ov  \mathfrak s(h_1+\gamma_1
)}{\boldsymbol {\pi}(h_2+\gamma_2)  \ov  \mathfrak s(h_2+\gamma_2
)}  e^{-{1\ov \be}\left( \Vert h_1+\gamma_1\Vert^{2} + \Vert h_2+\gamma_2\Vert^{2}\right)} 
\\ 
 \times \int_{\TT} \bm\pi(h)^{2}e^{-{2\ov \be}\Vert h\Vert^{2}}
\left(\prod_{p=1,2} \int_{G/\TT}   e^{{2\ov \be}(h_p+\gamma_p,g_{p,0}h g_{p,0}^{-1})}dg_{p,0}\right)d\bm{t},
\end{gathered}
\eeq
where the extra factor $2^r$ comes from \eqref{eq:GSMfermzero}, since the path integral measure for fermions is scale-invariant (see \cite[Remark 5]{Choi:2021yuz}). Here we effectively reduced the integration range of the bosonic zero modes $g_{p,0}$ to be $G/\TT$ with an overall factor of $\text{vol}(\TT)$  by using invariance of the classical action \eqref{eq:Scl} under the right quotient over $\TT$, i.e. $g_p\mapsto g_p \bm{t}$ for $\bm{t}\in \TT$.

Now we can use the Harish-Chandra formula \eqref{eq:HCAC} to compute the integrals over $G/\TT$, which gives
\beq \label{eq:inter}
\begin{gathered}
\Tr_{L^2(G)} \left[L_{g_l} R_{g_r^{-1}} e^{-{1\ov 2}\be {\Delta_G}} \right] 
={ (-1)^{\,\text{dim}(G/\TT)\ov 2}\text{vol}(\TT) e^{{1\ov 2}\be\ex{\rho,\rho}} \ov |W|(\pi\be) ^r} \\
\times \sum_{\gamma_1,\gamma_{2}\in\Gamma} \sum_{w_1,w_{2}\in W}\frac{ \e(w_1)\e(w_2)}{\mathfrak{s}(h_1+\gamma_1)\mathfrak{s}(h_2+\gamma_2)}  e^{-{1\ov \be}\left(\Vert h_1+\gamma_1\Vert^{2} + \Vert h_2+\gamma_2\Vert^{2}\right)} 
\\  \times\int _{\TT} 
e^{-{2\ov \be}\Vert h\Vert^{2}+{2\ov \be}(w_1(h_1+\gamma_1)+w_{2}(h_{2}+\gamma_{2}),h)}d\bm{t}.
\end{gathered}
\eeq
Completing the square in the total exponential factor in \eqref{eq:inter}, we obtain
\beq\label{eq:exp}
\begin{gathered}
-\frac{2}{\beta}\left\Vert h-\frac{w_1(h_1+\gamma_1)+ w_{2}(h_{2}+\gamma_{2})}{2}\right\Vert^{2} \\ - \frac{1}{2\beta}\Vert w_{1}h_{1}+w_{1}\gamma_{1}-w_{2}h_{2}-w_{2}\gamma_{2}\Vert^{2}.
\end{gathered}
\eeq

To simplify, for fixed $w_{1}, w_{2}\in W$, the summation over $\Gamma\times\Gamma$ in \eqref{eq:inter}, we use another basic fact.
\begin{lemma} \label{Lemma2} let $\rho$ be the Weyl vector. Then
$e^{\ex{\rho,w \gamma}}=e^{\ex{\rho, \gamma}} $ for all $w\in W$ and $\gamma\in \Gamma$.
\end{lemma}
\begin{proof}
Since $\ex{\rho,w \gamma}=\ex{w^{-1}\rho, \gamma}$, it is sufficient to verify this formula for a Weyl reflection $w_{\alpha}$, where $\alpha$ is a simple root.  
Denoting by $\alpha^{\vee}$ the corresponding coroot, we have
$$
w_{\al} \rho=\rho- \ex{\rho,\al^\vee}\al=\rho-\al,
$$
since $\ex{\rho,\alpha^{\vee}}=1$.
Then by Lemma \ref{eq:pairingId1} we obtain
$$
e^{\ex{w_{\al} \rho, \gamma}}=e^{\ex{\rho-\al, \gamma}}=e^{\ex{\rho, \gamma}}.\qedhere
$$
\end{proof}

Using  \eqref{eq:WeyldenId} and Lemma \ref{Lemma2}, we obtain
$${1\ov \mathfrak s(h_1+\gamma_1)\frak{s}(h_{2} +\gamma_{2})}= e^{\ex{\rho,w_1 \gamma_1-w_2\gamma_2 } } {1\ov \mathfrak s(h_1)\mathfrak s(h_2)  }.$$
Now for fixed $w_1,w_2\in W$ we put
$$\tilde \gamma_1={w_1 \gamma_1+w_2\gamma_2\ov 2}, \quad \tilde \gamma_2 =w_1\gamma_1-w_2\gamma_2,$$
so for $w_{1}\gamma_{1}, w_{2}\gamma_{2}\in\Gamma$ we have
$w_{1}\gamma_{1}=\tilde \gamma_1 + \frac{1}{2}\tilde\gamma_{2}$ and $w_{2}\gamma_{2}=\tilde \gamma_1-\frac{1}{2}\tilde\gamma_{2}$.
Thus if $\tilde\gamma_{2}=2\lambda_{2} +\sigma$, where $\lambda_{2}\in\Gamma$ and $\sigma\in\Gamma/2\Gamma$, then necessarily $\tilde\gamma_{1}=\lambda_{1}+\frac{1}{2}\sigma$, where $\lambda_{1}\in\Gamma$, and we have an isomorphism
\beq\label{e:iso}
\Gamma\times\Gamma\simeq\bigsqcup_{\sigma\in\Gamma/2\Gamma}(\Gamma+\tfrac{1}{2}\sigma)\times (2\Gamma+\sigma).
\eeq
Thus we can rewrite the double sum over $\Gamma\times\Gamma$ in \eqref{eq:inter} as follows
$$\sum_{\gamma_{1}\in\Gamma}\sum_{\gamma_{2}\in\Gamma}f(\tilde{\gamma}_{1},\tilde\gamma_{2})=\sum_{\lambda_{2}\in 2\Gamma}\sum_{\sigma\in\Gamma/2\Gamma}\sum_{\lambda_{1}\in\Gamma}f(\lambda_{1}+\tfrac{1}{2}\sigma,\lambda_{2}+\sigma).$$
Combining the sum over $\lambda_{1}$ with the integral along the maximal torus, after change of variables $h\to h +(w_{1}h_{1}+w_{2}h_{2}-\sigma)/2$ we obtain a simple Gaussian integral over $\frak{t}\simeq\RR^{r}$,
\beq \label{eq:Tint}
 \sum_{\lambda_1\in \Gamma} \int_\TT e^{-{2\ov \be} \Vert h+\lambda_{1}\Vert^{2}}d\bm{t}= \left({\pi \be\ov 2}\right)^{r\ov 2},
\eeq
where we used our choice of the Haar measure on $G$.

Thus we obtain
\beq 
\begin{gathered}
\Tr_{L^2(G)} \left[L_{g_l} R_{g_r^{-1}} e^{-{1\ov 2}\be {\Delta_G}} \right] ={\text{vol}(\TT) e^{{1\ov 2}\be\ex{\rho,\rho}} \ov |W|(2\pi\be) ^{r/2} \mathfrak s(h_1)\mathfrak s(-h_2)} 
\\
 \times \sum_{\tilde \gamma_2\in\Gamma} 
\sum_{w_{1},w_{2}\in W} \e(w_1)\e(w_2) e^{\ex{\rho, \tilde \gamma_2 }} e^{-{1\ov2\be}\Vert w_1h_1-w_2h_2+\tilde \gamma_2\Vert^{2} }.
\end{gathered}
\eeq
Introducing new variables $w=w_1^{-1}w_2\in W$ and $\gamma=w_1^{-1}\tilde \gamma_2 \in\Gamma$, we obtain
\beq
\begin{gathered}
 \sum_{\tilde \gamma_2\in\Gamma} 
\sum_{w_1,w_{2}\in W} \e(w_1)\e(w_2) e^{\ex{\rho, \tilde \gamma_2 }} e^{-{1\ov2\be}\Vert w_1h_1-w_2h_2+\tilde \gamma_2\Vert^{2} }
\\=   \sum_{w_1\in W} \sum_{ (w,\gamma) \in W\times\Gamma} 
\e(w) e^{\ex{\rho, \gamma }} e^{-{1\ov2\be}\Vert h_1-wh_2+\gamma\Vert^{2} }
\\=|W| \sum_{ (w,\gamma) \in W\times\Gamma}  
\e(w)e^{\ex{\rho, \gamma}} e^{-{1\ov2\be}\Vert h_1-wh_2+\gamma\Vert^{2} }.
\end{gathered}
\eeq
This completes our second derivation of the generalized Frenkel trace formula \eqref{eq:Frenkel trace formula final} using the gauged sigma model picture.

\section{Outlook}

We conclude by outlining several future directions suggested by our results.

First, we emphasize again that the \emph{Frenkel trace formula} naturally bridges the nonlinear sigma model (NLSM) approach underlying the Eskin formula \cite{Choi:2021yuz} and the gauged sigma model approach to the Selberg trace formula \cite{Choi:2023pjn}. Since the Eskin formula arises as a special degenerate limit of the Frenkel trace formula, it is clear that the same gauged sigma model framework should provide an alternative and conceptually straightforward derivation of the Eskin trace formula.

In contrast, the Selberg trace formula presents a qualitatively different situation. In that case, we were unable to apply our new localization principle directly at the level of supersymmetric quantum mechanics with target a hyperbolic Riemann surface. Instead, the gauged sigma model description appears to be more powerful: its localization framework provides a unified semiclassical interpretation of the identity, hyperbolic, and elliptic contributions to the Selberg trace formula. Historically, semiclassical analyses of a particle on a Riemann surface have only transparently captured the hyperbolic sector \cite{gutzwiller1980classical,gutzwiller1985geometry}. Establishing a comparable semiclassical interpretation for all contributions directly within the particle-on-a-Riemann-surface picture therefore seems to be a necessary step toward realizing localization in the NLSM framework.

A second natural direction is the extension of the Frenkel trace formula to \emph{non-compact semisimple Lie groups}, such as $SL(2,\mathbb{R})$ or $SL(2,\mathbb{C})$. The generalization of the Eskin formula to non-compact semisimple Lie groups—physically corresponding to the real-time quantum mechanical propagator between states related by exponentiated group elements $g=e^{h}$, $h\in\mathfrak g$—was first studied in \cite{Krausz:1997gw}, and later reformulated compactly using the localization principle of the present work in \cite{Choi:2023pjn}. In the Frenkel setting, however, additional structure arises because left and right twists may belong to distinct conjugacy classes of $G$. Understanding this richer structure may shed light on the nature of quantum mechanical propagation in non-compact Lie groups, particularly in so-called \emph{shadow regions}, where group elements are not connected to the identity by exponentiation (i.e.\ $G \neq e^{\mathfrak g}$), as occurs for example in $SL(2,\mathbb{R})$.

Another important generalization concerns \emph{higher-dimensional quantum field theories}, where the most natural framework is provided by two-dimensional sigma models with target space \(G\). In \cite{Murthy:2025ioh}, such generalizations were studied in the context of WZW conformal field theories with compact Lie groups, as well as non-compact targets such as \(SL(2,\mathbb{R})\) and \(H_3^+\). Since the (supersymmetric) quantum mechanics on \(G\) appearing in our localization framework can be viewed as a dimensional reduction of (supersymmetric) WZW models, it would be interesting to apply the gauged sigma model perspective developed here to gain further insight into the two-dimensional NLSM formulation, as well as to analyze coset conformal field theories.

It is also worth noting possible connections to related developments in two-dimensional conformal field theory and holography. Trace-formula-like structures have appeared in the study of AdS$_3$/CFT$_2$, including relations between modular invariance, spectral statistics, and random matrix theory, as well as in Gutzwiller-type trace formulas for two-dimensional CFT spectra \cite{DiUbaldo:2023qli}. While these works address rather different physical questions, the appearance of trace-formula-like structures in both settings highlights intriguing structural parallels, whose precise relationship to the framework developed here would be interesting to explore.

Finally, it would be interesting to explore the \emph{holographic interpretation} of the supersymmetric deformations introduced in this work, at least in the setting of a particle on a Lie group or supersymmetric WZW models. As discussed for example in \cite{Mertens:2018fds}, a particle on a Lie group may be interpreted as a holographic dual of BF theory with suitable boundary conditions, closely related to the Chern-Simons/WZW correspondence \cite{Witten:1988hf,Elitzur:1989nr}. Understanding how our localization framework fits into this holographic picture could provide a deeper conceptual understanding of the trace formulas studied here.

\bibliographystyle{amsplain}
\bibliography{Ref.bib}

\end{document}